\documentclass{stacs_proc}
\usepackage{graphicx}
\usepackage{amsmath}
\usepackage{amsfonts}
\usepackage{amssymb}
\usepackage{multirow}

\newcommand{\mcsp}{\mathtt{mcsp}}

\newcommand{\masp}{\mathtt{masp}}

\begin{document}

\title[Fixed parameter polynomial time algorithms for MASP and MCSP]
{Fixed Parameter Polynomial Time Algorithms for
Maximum Agreement and Compatible Supertrees}

\author[]{Hoang}{Viet Tung Hoang $^{1, 2}$}
\address[nus]{$^1$ Department of Computer Science, National University of Singapore}
\email{{hoangvi2,ksung}@comp.nus.edu.sg}

\author[]{Sung}{Wing-Kin Sung $^{1, 2}$}
\address[gis]{$^2$ Genome Institute of Singapore}

\subjclass{Algorithms, Biological computing}

\begin{abstract}
Consider a set of labels $L$ and a set of trees
${\mathcal T} = \{ {\mathcal T}^{(1)}, {\mathcal T}^{(2)}, \ldots, 
{\mathcal T}^{(k)} \}$ where 
each tree ${\mathcal T}^{(i)}$ is distinctly leaf-labeled by some subset of $L$.
One fundamental problem is to find the biggest tree (denoted as supertree)
to represent $\mathcal T$
which minimizes the disagreements with the trees in 
${\mathcal T}$ under certain criteria.
This problem finds applications in phylogenetics, database, and data mining.
In this paper, we focus on two particular supertree problems,
namely, the maximum agreement supertree problem (MASP)
and the maximum compatible supertree problem (MCSP).
These two problems are known to be NP-hard for $k \geq 3$.
This paper gives the first polynomial time algorithms
for both MASP and MCSP when 
both $k$ and the maximum degree $D$ of the trees are constant.
\end{abstract}

\maketitle

\stacsheading{2008}{361-372}{Bordeaux}
\firstpageno{361}


\section{Introduction}\label{intro}

Given a set of labels $L$ and
a set of unordered trees
${\mathcal T} = \{ {\mathcal T}^{(1)}, \ldots, {\mathcal T}^{(k)} \}$ 
where each tree ${\mathcal T}^{(i)}$ is distinctly leaf-labeled by some subset of $L$.
The supertree method tries to find a tree to represent all trees in
${\mathcal T}$ which minimizes the possible conflicts in the input trees.
The supertree method finds applications in 
phylogenetics, database, and data mining.
For instance, in the Tree of Life project~\cite{TOL},
the supertree method is the basic tool to infer 
the phylogenetic tree of all species.

Many supertree methods have been proposed in the literature
\cite{BN04,Gordon86,Guillemot07,Jansson05}.
This paper focuses on two particular supertree methods, namely
the Maximum Agreement Supertree (MASP)~\cite{Jansson05} 
and the Maximum Compatible Supertree (MCSP)~\cite{BN04}.
Both methods try to find a consensus tree with the largest number of leaves
which can represent all the trees in $\mathcal T$ under certain criteria.
(Please read Section~\ref{sec_definition} for the formal definition.) 

MASP and MCSP are known to be NP-hard as they 
are the generalization of the Maximum Agreement Subtree problem (MAST)
\cite{AK97,FPT95,KLST01} and
the Maximum Compatible Subtree problem (MCT) \cite{HJWZ96,GW01} respectively.
Jansson et al. \cite{Jansson05} proved that MASP remains NP-hard even if
every tree is a rooted triplet, i.e., a binary tree of $3$ leaves.
For $k=2$, Jansson et al. \cite{Jansson05} and 
Berry and Nicolas \cite{BN04}
proposed a linear time algorithm to 
transform MASP and MCSP for $2$ input trees 
to MAST and MCT respectively. 
For $k \geq 3$, positive results for computing
MASP/MCSP are reported only for rooted binary trees.
Jansson et al. \cite{Jansson05} gave an $O \bigl( k(2n)^{3k^2} \bigr)$ time 
solution to this problem.
Recently, Guillemot and Berry \cite{Guillemot07} further improve
the running time to $O(8^k n^k)$.

In general, the trees in $\mathcal T$ may not be binary nor rooted.
Hence, Jansson et al. \cite{Jansson05} posted an open problem 
and asked if MASP can be solved in polynomial time when
$k$ and the maximum degree of the trees in $\mathcal T$ are constant. 
This paper gives an affirmative answer to this question.
We show that both MASP and MCSP can be solved in polynomial time
when $\mathcal T$ contains constant number of bounded degree trees.
For the special case where the trees in $\mathcal T$ are
rooted binary trees, we show that
both MASP and MCSP can be solved in $O( 6^k n^k)$ time, 
which improves the previous best result.
Table~\ref{table:result} summarizes the previous and new results.

The rest of the paper is organized as follows. 
Section~\ref{sec_definition} gives the formal definition of the problems.
Then, Sections~\ref{mcsp} and \ref{mcsp_unroot} describe
the algorithms for solving MCSP for both rooted and unrooted cases.
Finally, Sections~\ref{masp} and \ref{masp_unroot} detail
the algorithms for solving MASP for both rooted and unrooted cases.
Proofs omitted due to space limitation will appear in the full
version of this paper.

\begin{table}
\begin{center}
\begin{tabular}{|l|lc|lc|} \hline
 & \bf Rooted & & \bf Unrooted & \\ \hline
MASP for $k$ trees of max degree $D$ & $O((kD)^{kD+3} (2n)^k)$ \phantom{\Large{I}} 
& $\dag$ & $O((kD)^{kD+3} (4n)^k)$ & $\dag$ \\ \hline
MCSP for $k$ trees of max degree $D$ & $O(2^{2kD} n^k)$ \phantom{\Large{I}} 
& $\dag$ & $O(2^{2kD} n^k)$ & $\dag$ \\ \hline
\multirow{3}{*}{MASP/MCSP for $k$ binary trees}
& $O\bigl(k (2n^2)^{3k^2}\bigr)$ \phantom{\LARGE{I}} & \cite{Jansson05} & & \\
& $O(8^k n^k)$ & \cite{Guillemot07} & & \\
& $O(6^k n^k)$ & $\dag$  & & \\ \hline
\end{tabular}
\end{center}
\caption{Summary of previous and new results ($\dag$ stands for new result).}
\label{table:result}
\end{table}

\section{Preliminary} \label{sec_definition}

A \emph{phylogenetic tree} is defined as 
an unordered  and distinctly leaf-labeled tree.
Given a phylogenetic tree $T$, the notation $L(T)$ 
denotes the leaf set of $T$, and the \emph{size} of $T$ 
refers to $|L(T)|$. 
For any label set $S$,  the \emph{restriction} of $T$
to $S$, denoted $T|S$, is a phylogenetic tree obtained from $T$ by 
removing all leaves in $L(T) - S$
and then suppressing all internal nodes of degree two. 
(See Figure~\ref{fig:refine_example} for an example of \emph{restriction}.)
For two phylogenetic trees $T$ and $T'$, we say that $T$ \emph{refines} $T'$, 
denoted $T \unrhd T'$, if $T'$ can be obtained by contracting
some edges of $T$. (See Figure~\ref{fig:refine_example} for an example
of \emph{refinement}.)

\begin{figure}[htbp]
	\centering
		\includegraphics[width=0.75\textwidth]{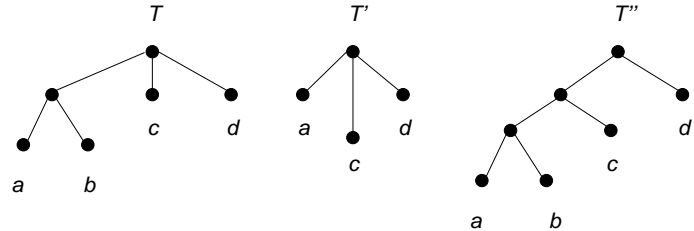}
	\caption{Three rooted trees. A tree $T$, a tree $T'$ 
	such that $T' = T \mid \{a, c, d\}$, and 
	a tree $T''$ such that $T'' \unrhd T$.}
	\label{fig:refine_example}
\end{figure}

\textbf{Maximum Compatible Supertree Problem}:
Consider a set of $k$ phylogenetic trees
${\mathcal T} = \{ {\mathcal T}^{(1)}, \ldots, {\mathcal T}^{(k)} \}$.
A \emph{compatible supertree}
of ${\mathcal T}$ is a tree $Y$
such that $Y|L({\mathcal T}^{(i)}) \unrhd {\mathcal T}^{(i)}|L(Y)$ for all $i \leq k$. 
The Maximum Compatible Supertree Problem (MCSP) is
to find a compatible supertree
with as many leaves as possible.
Figure~\ref{fig:supertree_example} shows an example 
of a compatible supertree
$Y$ of two rooted phylogenetic trees ${\mathcal T}^{(1)}$ and ${\mathcal T}^{(2)}$.
If all input trees have the same leaf sets, 
MCSP is referred as 
Maximum Compatible Subtree Problem (MCT).

\textbf{Maximum Agreement Supertree Problem:}
Consider a set of $k$ phylogenetic trees
${\mathcal T} = \{ {\mathcal T}^{(1)}, \ldots, {\mathcal T}^{(k)} \}$.
An \emph{agreement supertree}
of ${\mathcal T}$ is a tree $X$
such that $X|L({\mathcal T}^{(i)}) = {\mathcal T}^{(i)}|L(X)$ for all $i \leq k$.
The Maximum Agreement Supertree Problem (MASP) 
is to find an agreement supertree
with as many leaves as possible.
Figure~\ref{fig:supertree_example} shows an example 
of an agreement supertree $X$
of two rooted phylogenetic trees ${\mathcal T}^{(1)}$ and ${\mathcal T}^{(2)}$.
If all input trees have the same leaf sets, 
MASP is referred as 
Maximum Agreement Subtree Problem (MAST).

\begin{figure}[htbp]
	\centering
		\includegraphics[width=0.75\textwidth]{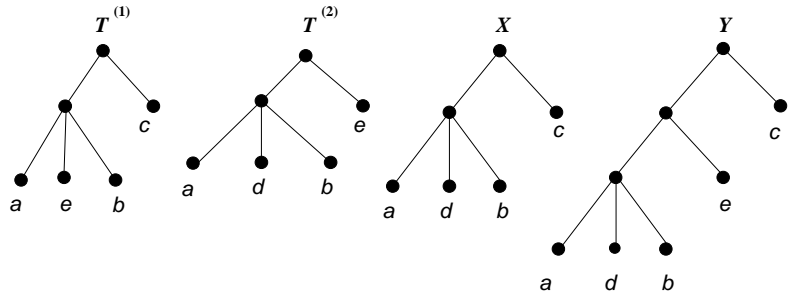}
	\label{fig:supertree_example}
	\caption{An agreement supertree $X$ and	a compatible supertree $Y$ of 
	$2$ rooted phylogenetic trees ${\mathcal T}^{(1)}$ 	and ${\mathcal T}^{(2)}$.}
\end{figure}

In the following discussion, for the set of phylogenetic trees
${\mathcal T} = \{ {\mathcal T}^{(1)}, \ldots, {\mathcal T}^{(k)} \}$, 
we denote $n = |\bigcup_{i=1..k} L({\mathcal T}^{(i)})|$, 
and $D$ stands for the maximum degree of the trees in ${\mathcal T}$.
We assume that none of the trees in $\mathcal T$ has an internal node of degree two, 
so that each tree contains at most $n - 1$ internal nodes. 
(If a tree ${\mathcal T}^{(i)}$ has some internal nodes
of degree two, we can replace it by ${\mathcal T}^{(i)} \mid L({\mathcal T}^{(i)})$ 
in linear time.)

\section{Algorithm for MCSP of rooted trees}\label{mcsp}

Let $\mathcal T$ be a set of $k$
rooted phylogenetic trees. 
This section presents a dynamic programming algorithm 
to compute the size of a maximum compatible supertree 
of $\mathcal T$ in  $O\left( 2^{2k D}n^k \right)$ time. 
The maximum compatible supertree can be obtained in the 
same asymptotic time bound by backtracking.

For every compatible supertree $Y$ of $\mathcal T$, 
there exists a binary tree that refines $Y$. 
This binary tree is also a compatible supertree of $\mathcal T$, 
and is of the same size as $Y$. 
Hence in this section, every compatible supertree is implicitly
assumed to be binary. 

\begin{definition}[Cut-subtree]
A \emph{cut-subtree} of a tree $T$ is either
an empty tree or 
a tree obtained by first selecting some subtrees attached to
the same internal node in $T$
and then connecting those subtrees by a common root.
\label{D:mcsp_cut-subtree}
\end{definition}

\begin{definition}[Cut-subforest]
Given a set of $k$ rooted (or unrooted) trees $\mathcal T$,
a \emph{cut-subforest} of $\mathcal T$ is a set 
${\mathcal A} = \{ {\mathcal A}^{(1)},  \ldots ,{\mathcal A}^{(k)} \}$, where 
${{\mathcal A}}^{(i)}$ is a cut-subtree of ${\mathcal T}^{(i)}$
and at least one element of ${\mathcal A}$ is not an empty tree. 
\label{D:mcsp_cut-subforest}
\end{definition}

\begin{figure}[htbp]
	\centering
		\includegraphics[width=0.7\textwidth]{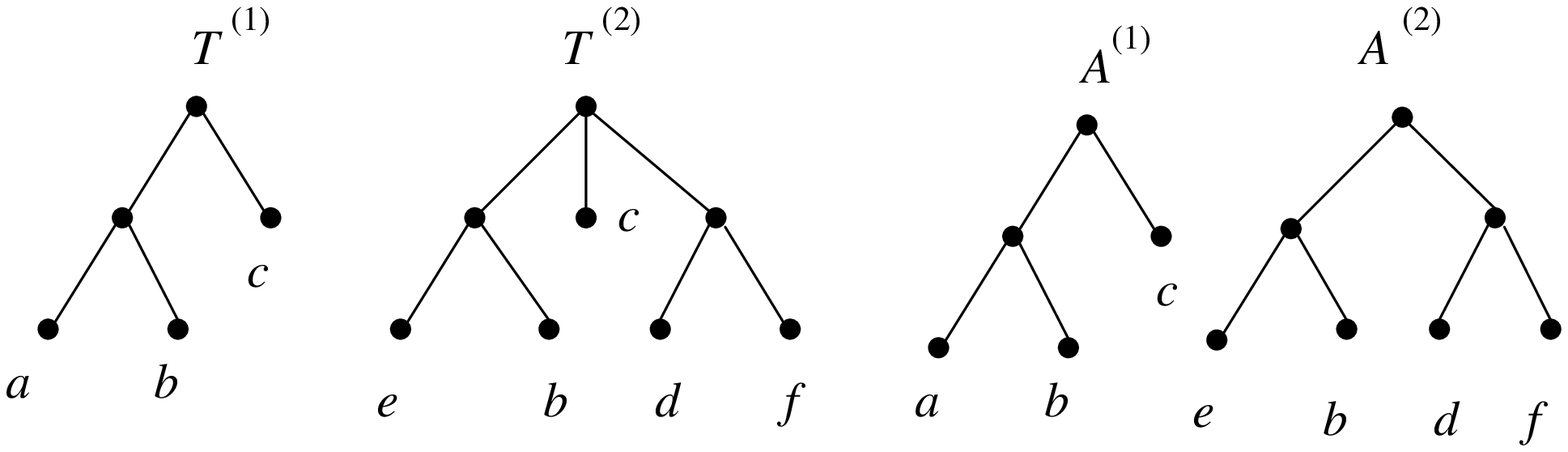}		
	\caption{A cut-subforest $\mathcal A$ of $\mathcal T$.}
	\label{fig:cut_subtree_example}
\end{figure}

For example, in Figure~\ref{fig:cut_subtree_example},
$\{ {\mathcal A}^{(1)}, {\mathcal A}^{(2)} \}$ is 
a cut-subforest of $\{ {\mathcal T}^{(1)}, {\mathcal T}^{(2)} \}$.
Let $\mathcal O$ denote the set of all possible cut-subforests of $\mathcal T$.

\begin{lemma}
There are $O\left(2^{k D}n^k\right)$ different cut-subforests of $\mathcal T$.
\label{L:mcsp_bound_position}
\end{lemma} 
\begin{proof}
We claim that each tree ${\mathcal T}^{(i)}$ contributes 
$2^D n$ or fewer cut-subtrees; therefore there are 
$O\left(2^{k D}n^k\right)$ cut-subforests of $\mathcal T$.
At each internal node $v$ of ${\mathcal T}^{(i)}$, 
since the degree of $v$ does not exceed $D$, we have 
at most $2^D$ ways of selecting the subtrees attached to $v$ 
to form a cut-subtree. Including the empty tree, 
the number of cut-subtrees in ${\mathcal T}^{(i)}$ cannot go beyond
$(n - 1)2^D + 1 < 2^D n$.
\end{proof}

Figure~\ref{fig:mcsp_substructure_example} demonstrates
that a compatible supertree of some cut-subforest ${\mathcal A}$ of $\mathcal T$
may not be a compatible supertree of $\mathcal T$.
To circumvent this irregularity, we define \emph{embedded supertree} as follows.

\begin{definition}[Embedded supertree]
For any cut-subforest ${\mathcal A}$ of $\mathcal T$,
a tree $Y$ is called an \emph{embedded supertree} of ${\mathcal A}$ if
$Y$ is a compatible supertree of ${\mathcal A}$, and
$L(Y) \cap L({\mathcal T}^{(i)}) \, \subseteq \, L({\mathcal A}^{(i)})$
for all $i \leq k$.
\label{D:mcsp_embed}
\end{definition}

\begin{figure}[htbp]
	\centering
		\includegraphics[width=1.0\textwidth]{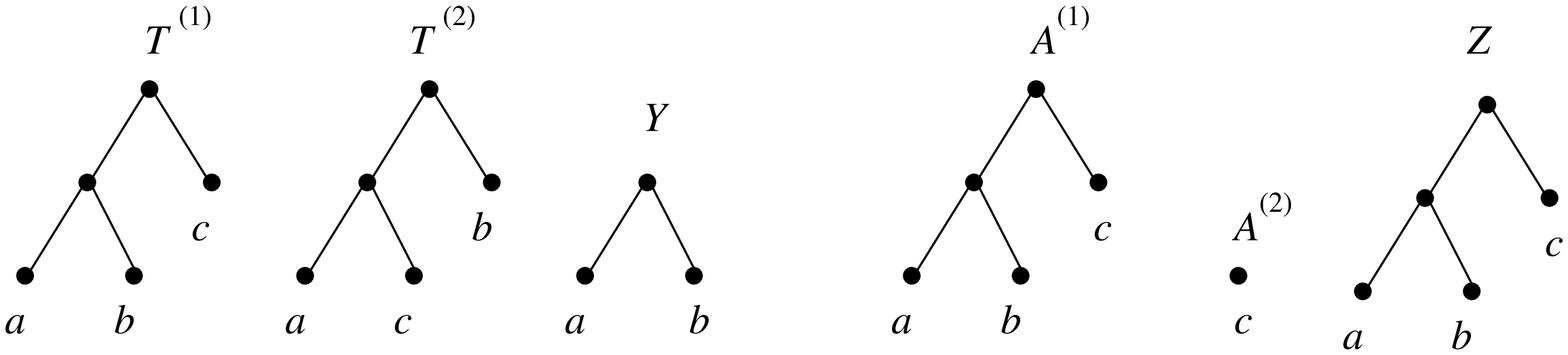}
	\caption{Consider ${\mathcal T} = \{{\mathcal T}^{(1)}, {\mathcal T}^{(2)} \}$
	and its cut-subforest ${\mathcal A} = \{ {\mathcal A}^{(1)}, {\mathcal A}^{(2)} \}$.
	Although $Z$ is a compatible supertree of $\mathcal A$, it is not 
	a compatible supertree of $\mathcal T$. The maximum compatible supertree of
	$\mathcal T$ is $Y$ that contains only $2$ leaves.}
	\label{fig:mcsp_substructure_example}
\end{figure}

Note that a compatible supertree of $\mathcal T$
is also an embedded supertree of $\mathcal T$. 
For each cut-subforest  ${\mathcal A}$ of $\mathcal T$, let $\mcsp({\mathcal A})$ 
denote the maximum size of embedded supertrees  of ${\mathcal A}$. 
Our aim is to compute $\mcsp({\mathcal T})$.
Below, we first define the recursive equation for computing
$\mcsp({\mathcal A})$ for all cut-subforests ${\mathcal A} \in {\mathcal O}$.
Then, we describe our dynamic programming algorithm.

We partition the cut-subforests in ${\mathcal O}$ into two classes.
A cut-subforest  ${\mathcal A}$ of $\mathcal T$ is \emph{terminal} if 
each element ${{\mathcal A}}^{(i)}$ is either an empty tree or 
a leaf of ${\mathcal T}^{(i)}$; it is called \emph{non-terminal}, otherwise.

For each terminal cut-subforest ${\mathcal A}$, 
let 
\begin{equation}
\Lambda({\mathcal A}) = \Bigl\{ l \in \bigcup_{j = 1..k} L( {\mathcal A}^{(j)}) 
\mid l \not\in \; L({\mathcal T}^{(i)}) - L({\mathcal A}^{(i)})
\mbox{ for } i = 1,2, \ldots, k \Bigr\} \enspace .
\label{eq:Lambda_function}
\end{equation}
For example,  with $\mathcal T$ in Figure~\ref{fig:supertree_example},
if ${\mathcal A}^{(1)}$ and ${\mathcal A}^{(2)}$ are leaves labeled by $a$
and $d$ respectively then $\Lambda({\mathcal A}) = \{ d \}$. 
In Lemma~\ref{L:mcsp_terminal}, we show that 
$\mcsp({\mathcal A}) = \left| \Lambda({\mathcal A}) \right|$. 

\begin{lemma}
If ${\mathcal A}$ is a terminal cut-subforest then 
$\mcsp({\mathcal A}) = \left| \Lambda({\mathcal A}) \right|$.
\label{L:mcsp_terminal}
\end{lemma}
\begin{proof}
Consider any embedded supertree $Y$ of ${\mathcal A}$. 
By Definition~\ref{D:mcsp_embed}, 
every leaf of $Y$ belongs to $\Lambda({\mathcal A})$. 
Hence the value $\mcsp({\mathcal A})$ does not exceed $\left| \Lambda({\mathcal A}) \right|$. 

It remains to give an example of some embedded supertree of ${\mathcal A}$
whose leaf set is $\Lambda({\mathcal A})$. Let $C$ be a rooted caterpillar
\footnote{A rooted caterpillar is a rooted, unordered, and 
distinctly leaf-labeled binary tree where every internal
node has at least one child that is a leaf.}
whose leaf set is $\Lambda({\mathcal A})$. The definition of $\Lambda({\mathcal A})$
implies that $L(C) \cap L({\mathcal T}^{(i)}) \,
\subseteq \, L\bigl({\mathcal A}^{(i)} \bigr)$ for every $i \leq k$. 
Since each ${\mathcal A}^{(i)}$
has at most one leaf, it is straightforward 
that $C$ is a compatible supertree of ${\mathcal A}$. 
Hence $C$ is the desired example. 
\end{proof}

\begin{definition}[Bipartite]
Let ${\mathcal A}$ be a cut-subforest  of $\mathcal T$.
We say that the cut-subforests 
${\mathcal A}_L$ and ${\mathcal A}_R$ \emph{bipartition}
${\mathcal A}$ if for every $i \leq k$, 
the trees ${\mathcal A}_L^{(i)}$ and  ${\mathcal A}_R^{(i)}$ can be obtained by
$(1)$ partitioning the subtrees attached to the root of ${\mathcal A}^{(i)}$
into two sets $S^{(i)}_L$ and $S^{(i)}_R$; and
$(2)$ connecting the subtrees in $S^{(i)}_L$ (resp. $S^{(i)}_R$) 
by a common root to form ${\mathcal A}_L^{(i)}$ (resp. ${\mathcal A}_R^{(i)}$).
\label{D:mcsp_bipartite}
\end{definition}

Figure~\ref{fig:mcsp_bipartite_example} shows an example of 
the preceding definition.
For each non-terminal cut-subforest $\mathcal A$, we compute $\mcsp({\mathcal A})$ 
based on the mcsp values of ${\mathcal A}_L$ and ${\mathcal A}_R$
for each bipartite $({\mathcal A}_L, {\mathcal A}_R)$ of ${\mathcal A}$. 
More precisely, we prove that
\begin{equation}
\mcsp({\mathcal A}) = \max \{ \mcsp({\mathcal A}_L) + \mcsp({\mathcal A}_R) \mid 
{\mathcal A}_L \mbox{ and } {\mathcal A}_R \mbox{ bipartition } {\mathcal A} \} \enspace .
\label{eq:mcsp_recurrent_formula}
\end{equation}
The identity \eqref{eq:mcsp_recurrent_formula} is then established by 
Lemmas~\ref{L:mcsp_sufficient} and \ref{L:mcsp_necessary}. 

\begin{figure}[htbp]
	\centering
		\includegraphics[width=0.85\textwidth]{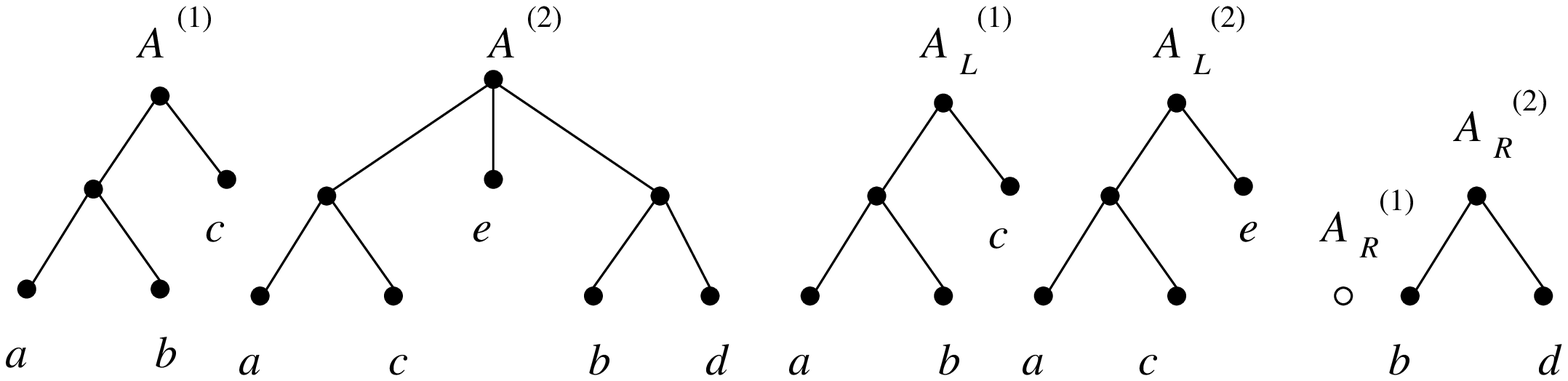}
	\caption{A bipartite $({\mathcal A}_L, {\mathcal A}_R)$ of 
	a cut-subforest $\mathcal A$. The empty tree is represented
	by a white circle.}
	\label{fig:mcsp_bipartite_example}
\end{figure}

\begin{lemma}
Consider a bipartite $({\mathcal A}_L, {\mathcal A}_R)$ of some cut-subforest ${\mathcal A}$
of $\mathcal T$.
If $Y_L$ and $Y_R$ are embedded supertrees of 
${\mathcal A}_L$ and ${\mathcal A}_R$ respectively
then $Y$ is an embedded supertree of ${\mathcal A}$, where $Y$ is
formed by connecting $Y_L$ and $Y_R$ to a common root.
\label{lem_embed_prop1}
\end{lemma}

\begin{lemma}
Let ${\mathcal A}$ be a cut-subforest  of $\mathcal T$.
If $({\mathcal A}_L, {\mathcal A}_R)$ is a bipartite of ${\mathcal A}$ then 
$\mcsp({\mathcal A}) \geq \mcsp({\mathcal A}_L) + \mcsp({\mathcal A}_R)$.
\label{L:mcsp_sufficient}
\end{lemma}
\begin{proof}
Consider an embedded supertree $Y_L$ of
${\mathcal A}_L$ such that $|L(Y_L)| = \mcsp({\mathcal A}_L)$. 
Define $Y_R$ for ${\mathcal A}_R$ similarly.
Let $Y$ be a tree
formed by connecting $Y_L$ and $Y_R$ with a common root.
Note that $Y$ is of size $\mcsp({\mathcal A}_L)+ \mcsp({\mathcal A}_R)$.
By Lemma~\ref{lem_embed_prop1}, 
$Y$ is an embedded supertree of ${\mathcal A}$ and
hence the lemma follows.
\end{proof}

\begin{lemma}
Given a cut-subforest ${\mathcal A}$ of $\mathcal T$, 
let $Y$ be a binary embedded supertree of $\mathcal A$ 
with left subtree $Y_L$ and right subtree $Y_R$.
There exists a bipartite $({\mathcal A}_L, {\mathcal A}_R)$ of ${\mathcal A}$
such that either $(i)$ $Y$ is an embedded supertree of ${\mathcal A}_L$;
or $(ii)$ $Y_L$ and $Y_R$ are 
embedded supertrees of ${\mathcal A}_L$ and ${\mathcal A}_R$ respectively.
\label{lem_embedded_prop2}
\end{lemma}

\begin{lemma}
For each non-terminal cut-subforest  ${\mathcal A}$ of $\mathcal T$, there exists
a bipartite $({\mathcal A}_L, {\mathcal A}_R)$ of ${\mathcal A}$ such that 
$\mcsp({\mathcal A}) \leq \mcsp({\mathcal A}_L) + \mcsp({\mathcal A}_R)$.
\label{L:mcsp_necessary}
\end{lemma}
\begin{proof}
Let $Y$ be a binary embedded supertree of ${\mathcal A}$ 
such that $\left| L(Y) \right| = \mcsp({\mathcal A})$. 
By Lemma~\ref{lem_embedded_prop2},
there exists a bipartite $({\mathcal A}_L, {\mathcal A}_R)$ of ${\mathcal A}$
such that either
(1) $Y$ is an embedded supertree of ${\mathcal A}_L$; or
(2) $Y_L$ and $Y_R$ are 
embedded supertrees of ${\mathcal A}_L$ and ${\mathcal A}_R$ respectively, 
where $Y_L$ is the left subtree and $Y_R$ is the right subtree of $Y$. 
In both cases, $|L(Y)| \leq \mcsp({\mathcal A}_L) + \mcsp({\mathcal A}_R)$.
Then the lemma follows.
\end{proof}

The above discussion then leads to Theorem~\ref{T:mcsp_formula}. 

\begin{theorem}
For every cut-subforest  ${\mathcal A}$ of $\mathcal T$, 
the value $\mcsp({\mathcal A})$ equals to
$$
\left\{
\begin{array}{lr}
\left| \Lambda({\mathcal A}) \right|,  \mbox{ if } {\mathcal A} \mbox{ is terminal,}\\
\max \{ \mcsp({\mathcal A}_L) + \mcsp({\mathcal A}_R) \mid {\mathcal A}_L \mbox{ and } {\mathcal A}_R \mbox{ bipartition } {\mathcal A} \}, \mbox{ otherwise} \enspace .
\end{array}
\right.
$$
\label{T:mcsp_formula}
\end{theorem}

We define an ordering of the cut-subforests in $\mathcal O$ as follows.
For any cut-subforests ${\mathcal A}_1, {\mathcal A}_2$ in $\mathcal O$,
we say that ${\mathcal A}_1$ is smaller than  ${\mathcal A}_2$ 
if ${\mathcal A}_1^{(i)}$ is a cut-subtree of ${\mathcal A}_2^{(i)}$ for $i=1,2,\ldots, k$.
Our algorithm enumerates ${\mathcal A} \in {\mathcal O}$ in topologically increasing order
and computes $\mcsp({\mathcal A})$ based on Theorem~\ref{T:mcsp_formula}.
Theorem~\ref{T:mcsp_complexity}
states the complexity of our algorithm. 

%
%
%
%

\begin{theorem}
A maximum compatible supertree of $k$ rooted phylogenetic trees
can be obtained in $O\left(2^{2k D} n^k\right)$ time .
\label{T:mcsp_complexity}
\end{theorem}
\begin{proof}
Testing if a cut-subforest is terminal takes $O(k)$ times, 
and each terminal cut-subforest  ${\mathcal A}$ then requires $O(k^2)$ time 
for the computation of $\Lambda({\mathcal A})$. In view of 
Lemma~\ref{L:mcsp_bound_position}, it suffices to show that 
each non-terminal cut-subforest  ${\mathcal A}$ has $O(2^{k D})$ bipartites. 
This result follows from the fact that for each $i \leq k$, 
there are at most $2^D$ ways to partition the set of
the subtrees attached to the root of  ${\mathcal A}^{(i)}$. 
\end{proof}

In the special case where every tree ${\mathcal T}^{(i)}$ is binary, 
Theorem~\ref{T:mcsp_binary} shows that our algorithm actually
has a better time complexity.
Note that the concepts of agreement supertree and compatible supertree 
will coincide for binary trees.
Hence, our algorithm
improves the $O\left(8^k n^k\right)$-time algorithm in \cite{Guillemot07}
for computing maximum agreement supertree of $k$ rooted binary trees. 

\begin{theorem}
If every tree in $\mathcal T$ is binary,
a maximum compatible supertree (or a maximum agreement supertree)
can be computed in $O\left(6^k n^k \right)$ time.
\label{T:mcsp_binary}
\end{theorem}
\begin{proof}
We claim that the processing of non-terminal cut-subforests of $\mathcal T$
requires $O\left(6^k n^k \right)$ time. The argument in the 
proof of Theorem~\ref{T:mcsp_complexity} tells that 
the remaining computation runs within the same 
asymptotic time bound. 
Consider an integer $r \in \{0, 1, \ldots, k \}$. We shall be dealing with
a cut-subforest  ${\mathcal A}$ such that there are exactly $r$ cut-subtrees 
${\mathcal A}^{(i)}$ whose roots are internal nodes of ${\mathcal T}^{(i)}$. 
The key of this proof is to show that 
the number of those cut-subforests does not exceed
$\left(
\begin{array}{c}
k \\
r
\end{array}
\right)(n-1)^r(n + 1)^{k - r}$,
and the running time for each cut-subforest  is 
$O\left(4^r 2^{k - r}\right)$. Hence, the total running time 
for all non-terminal cut-subforests is 
$$
\sum_{r = 0}^k \left(
\begin{array}{c}
k \\
r
\end{array}
\right)(n-1)^r(n + 1)^{k - r} O\left(4^r 2^{k - r}\right) 
= O\left(6^k n^k\right).
$$

We can count the number of the specified cut-subforests ${\mathcal A}$ as follows. 
First there are 
$\left(
\begin{array}{c}
k \\
r
\end{array}
\right)$ 
options for $r$ indices $i$ such that the roots of cut-subtrees 
${\mathcal A}^{(i)}$ are internal nodes of ${\mathcal T}^{(i)}$. 
For those cut-subtrees, we then appoint one of the $(n - 1)$
or fewer internal nodes of 
${\mathcal T}^{(i)}$ to be the root node of ${\mathcal A}^{(i)}$. 
Every other cut-subtree of $\mathcal A$ 
is a leaf or the empty tree, and then can be
determined from at most $n + 1$ alternatives. 
Multiplying those possibilities gives us 
the bound stipulated in the preceding paragraph. 

It remains to estimate the running time for each specified cut-subforest  ${\mathcal A}$. 
This task requires us to bound the number of bipartites of each cut-subforest. 
If the root $v$ of ${\mathcal A}^{(i)}$ is an internal node of ${\mathcal T}^{(i)}$ then 
${\mathcal A}^{(i)}$ contributes $4$ or fewer ways of partitioning the set 
of the subtrees attached to $v$. 
Otherwise,  we have at most $2$ ways of partitioning this set.  
Hence ${\mathcal A}$ owns at most $4^r 2^{k - r}$ bipartites, 
and this completes the proof. 
\end{proof}

\section{Algorithm for MCSP of unrooted trees}\label{mcsp_unroot}

Let ${\mathcal T}$ be a set of $k$ unrooted phylogenetic trees. 
This section extends the algorithm in Section~\ref{mcsp}
to find the size of a maximum compatible supertree of $\mathcal T$. 
The maximum compatible supertree can be obtained by backtracking. 
Surprisingly, the extended algorithm for unrooted trees 
runs within the same asymptotic time bound as the 
original algorithm for rooted trees. 

We will follow the same approach as Section~\ref{mcsp}, i.e., 
for each cut-subforest $\mathcal A$ of $\mathcal T$, we find an 
embedded supertree of $\mathcal A$ of maximum size. 
Definitions~\ref{D:mcsp_cut-subtree}, \ref{D:mcsp_cut-subforest},
and \ref{D:mcsp_embed} for cut-subforest and embedded supertree 
in the previous section are still valid for unrooted trees. 
Notice that although $\mathcal T$ is the set of unrooted trees, 
each cut-subforest $\mathcal A$ of $\mathcal T$ consists of 
rooted trees. (See Figure~\ref{fig:mcsp_unroot_example} 
for an example of cut-subforest for unrooted trees.) 
Hence we can use the algorithm in Section~\ref{mcsp}
to find the maximum embedded supertree of $\mathcal A$. We then 
select the biggest tree $T$ among those maximum embedded supertrees
for all cut-subforests of $\mathcal T$, and unroot $T$ to obtain the
maximum compatible supertree of $\mathcal T$. 

\begin{figure}[htbp]
	\centering
		\includegraphics[width=0.9\textwidth]{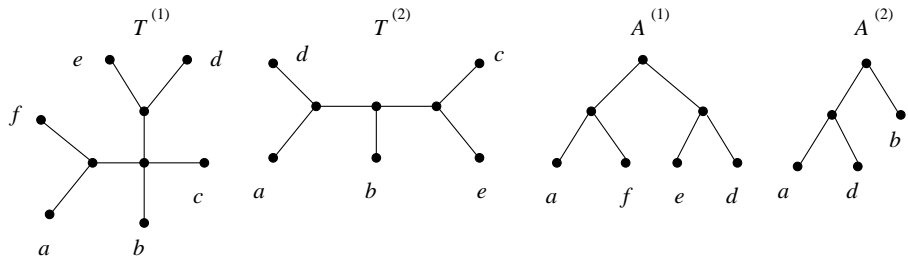}
	\caption{The set of rooted trees ${\mathcal A} = \{{\mathcal A}^{(1)}, 
	{\mathcal A}^{(2)} \}$ is a cut-subforest of 
	${\mathcal T} = \{ {\mathcal T}^{(1)}, {\mathcal T}^{(2)} \}$.}
	\label{fig:mcsp_unroot_example}
\end{figure}

Theorem~\ref{T:unroot_mcsp_complexity} shows that 
the extended algorithm has the same asymptotic time bound 
as the algorithm in Section~\ref{mcsp}. 

\begin{theorem}
We can find a maximum compatible supertree of $k$
unrooted phylogenetic trees in $O\left( 2^{2k D}n^k \right)$ time. 
\label{T:unroot_mcsp_complexity}
\end{theorem}
\begin{proof}
Using a similar proof as Lemma~\ref{L:mcsp_bound_position}, we can prove that 
there are $O\left( 2^{k D} n^k \right)$ cut-subforests of $\mathcal T$. 
As given in the proof of Theorem~\ref{T:mcsp_complexity}, 
finding the maximum embedded supertrees of each cut-subforest
takes $O(2^{k D})$ time. Hence the extended algorithm runs
within the specified time bound. 
\end{proof}

\vskip-0.3cm
\section{Algorithm for MASP of rooted trees}\label{masp}

Let $\mathcal T$ be a set of $k$
rooted phylogenetic trees. 
This section presents a dynamic programming algorithm 
to compute the size of a maximum agreement supertree 
of $\mathcal T$ in  $O\left( (k D)^{k D + 3} (2n)^k \right)$ time. 
The maximum agreement supertree can be obtained in the 
same asymptotic time bound by backtracking. 

The idea here is similar to that of Section~\ref{mcsp}.
However, while we can assume that compatible supertrees
are binary, the maximum degree of agreement supertrees
can grow up to $k D$. It is the reason why we have the factor
$O( (k D)^{k D + 3} )$ in the complexity. 

\begin{definition}[Sub-forest]
Given a set of $k$ rooted trees $\mathcal T$,
a \emph{sub-forest} of $\mathcal T$ is a set 
${\mathcal A} = \{ {\mathcal A}^{(1)},  \ldots ,{\mathcal A}^{(k)} \}$, where 
each ${{\mathcal A}}^{(i)}$ is either an empty tree or 
a complete subtree rooted at some node of ${\mathcal T}^{(i)}$, 
and at least one element of ${\mathcal A}$ is not an empty tree. 
\label{D:masp_sub-forest}
\end{definition}

Notice that the definition of sub-forest does not coincide
with the concept of cut-subforest in 
Definition~\ref{D:mcsp_cut-subforest} of Section~\ref{mcsp}. 
For example, the cut-subforest $\mathcal A$ 
in Figure~\ref{fig:cut_subtree_example} is not
a sub-forest of $\mathcal T$, because ${\mathcal A}^{(2)}$ 
is not a complete subtree rooted at some node of ${\mathcal T}^{(2)}$. 
Let $\mathcal O$ denote the set of all possible sub-forests of $\mathcal T$.
Then $\left| \mathcal O \right| = O\left( (2n)^k \right)$. 

\begin{definition}[Enclosed supertree]
For any sub-forest ${\mathcal A}$ of $\mathcal T$,
a tree $X$ is called an \emph{enclosed supertree} of ${\mathcal A}$ if
$X$ is an agreement supertree of ${\mathcal A}$, and
$L(X) \cap L({\mathcal T}^{(i)}) \, \subseteq \, L({\mathcal A}^{(i)})$
for all $i \leq k$.
\label{D:masp_enclose}
\end{definition}

For each sub-forest  ${\mathcal A}$ of $\mathcal T$, let $\masp({\mathcal A})$ 
denote the maximum size of enclosed supertrees  of ${\mathcal A}$. 
We use a similar approach as Section~\ref{mcsp}, i.e., 
we compute $\masp({\mathcal A})$ for all ${\mathcal A} \in {\mathcal O}$, 
and $\masp({\mathcal T})$ is the size of a maximum agreement
supertree of $\mathcal T$. 
We partition the sub-forests in $\mathcal O$ to two classes.
A sub-forest ${\mathcal A}$ is \emph{terminal} if each ${\mathcal A}^{(i)}$
is either an empty tree or a leaf. Otherwise, ${\mathcal A}$ is called 
\emph{non-terminal}. 

Notice that for terminal sub-forest, the definition of 
enclosed supertree coincides with the concept of embedded supertree 
in Definition~\ref{D:mcsp_embed} of Section~\ref{mcsp}. 
Then by Lemma~\ref{L:mcsp_terminal}, we have $\masp({\mathcal A}) 
= |\Lambda({\mathcal A})|$. (Please refer to the formula
\eqref{eq:Lambda_function}
in the paragraph preceding
Lemma~\ref{L:mcsp_terminal} for the definition of function $\Lambda$.)

\begin{definition}[Decomposition]
Let $\mathcal A$ be a sub-forest of $\mathcal T$. We say that 
sub-forests ${\mathcal B}_1, \ldots, {\mathcal B}_d$ (with $d \geq 2$)
\emph{decompose} $\mathcal A$ if for all $i \leq k$, 
either $(i)$ Exactly one of ${\mathcal B}_1^{(i)}, \ldots, {\mathcal B}_d^{(i)}$
is isomorphic to ${\mathcal A}^{(i)}$ while the others are empty trees; 
or $(ii)$ There are at least $2$ nonempty trees in 
${\mathcal B}_1^{(i)}, \ldots, {\mathcal B}_d^{(i)}$, and all those
nonempty trees are isomorphic to 
pairwise distinct subtrees attached to the root of ${\mathcal A}^{(i)}$. 
\label{D:masp_decomposition}
\end{definition}

\vskip-1cm

\begin{figure}[htbp]
	\centering
		\includegraphics[width=1.00\textwidth]{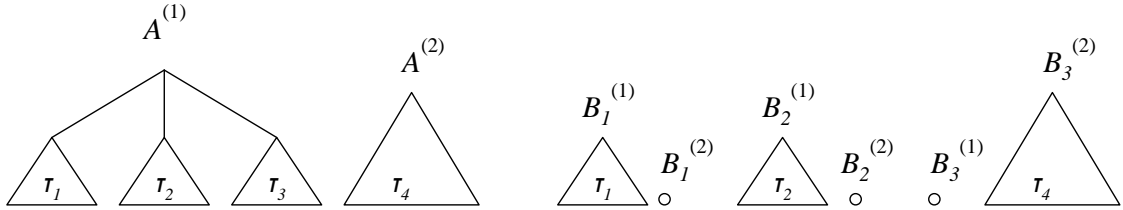}
	\caption{A decomposition $({\mathcal B}_1, {\mathcal B}_2, {\mathcal B}_3)$ of 
	a sub-forest ${\mathcal A}$. 
	The empty trees are represented by white circles.}
	\label{fig:masp_decompose_example}
\end{figure}

Figure~\ref{fig:masp_decompose_example} illustrates the 
concept of decomposition. 
For each sub-forest ${\mathcal A}$ of $\mathcal T$, we will prove that
\begin{equation}
\masp({\mathcal A}) = \max \{ \masp({\mathcal B}_1) + \ldots + \masp({\mathcal B}_d)
\mid {\mathcal B}_1, \ldots, {\mathcal B}_d \mbox{ decompose } {\mathcal A} \} \enspace .
\label{eq:masp_recurrent_formula}
\end{equation}
The identity \eqref{eq:masp_recurrent_formula} is then established
by Lemmas~\ref{L:masp_sufficient} and ~\ref{L:masp_necessary}.

\begin{lemma}
Suppose $({\mathcal B}_1, \ldots, {\mathcal B}_d)$ is  a decomposition 
of some sub-forest ${\mathcal A}$ of $\mathcal T$. Let $\tau_1, \ldots, \tau_d$
be some enclosed supertrees of ${\mathcal B}_1, \ldots, {\mathcal B}_d$ respectively, 
and let $X$ be the tree obtained by connecting $\tau_1, \ldots, \tau_d$
to a common root. Then, $X$ is an enclosed supertree of ${\mathcal A}$.  
\label{P:masp_prop_1}
\end{lemma}

\begin{lemma}
If $({\mathcal B}_1, \ldots, {\mathcal B}_d)$ is a decomposition of
a sub-forest ${\mathcal A}$ of $\mathcal T$ then
$\masp({\mathcal A}) \geq \masp({\mathcal B}_1) +  \ldots + \masp({\mathcal B}_d)$. 
\label{L:masp_sufficient}
\end{lemma}
\begin{proof}
For each ${\mathcal B}_j$, let $\tau_j$ be an enclosed supertree 
of ${\mathcal B}_j$ such that $|L(\tau_j)| = \masp({\mathcal B}_j)$. 
Let $X$ be the tree obtained by connecting $\tau_1, \ldots, \tau_d$
to a common root. By Lemma~\ref{P:masp_prop_1}, 
$X$ is an enclosed supertree of ${\mathcal A}$. Hence 
$|L(\tau_1)| + \ldots + |L(\tau_d)| = |L(X)| \leq \masp({\mathcal A})$. 
\end{proof}

\begin{lemma}
Let  $X$ be an enclosed supertree of some sub-forest 
${\mathcal A}$ of $\mathcal T$, and let $\tau_1, \ldots, \tau_d$
be all subtrees attached to the root of $X$. 
Then either $(i)$ There is a decomposition
$({\mathcal B}_1, {\mathcal B}_2)$ of ${\mathcal A}$ such that 
$X$ is an enclosed supertree of ${\mathcal B}_1$; or    
$(ii)$ There is a decomposition $({\mathcal B}_1, \ldots, {\mathcal B}_d)$
of ${\mathcal A}$ such that each $\tau_j$ is an enclosed supertree of ${\mathcal B}_j$. 
\label{P:masp_prop_2}
\end{lemma}

\begin{lemma}
For each non-terminal sub-forest ${\mathcal A}$ of $\mathcal T$, there 
is a decomposition $({\mathcal B}_1, \ldots, {\mathcal B}_d)$ of 
${\mathcal A}$ such that 
$\masp({\mathcal A}) \leq \masp({\mathcal B}_1) + \ldots + \masp({\mathcal B}_d)$
\label{L:masp_necessary}
\end{lemma}
\begin{proof}
Let $X$ be an enclosed supertree of ${\mathcal A}$ 
such that $|L(X)| = \masp({\mathcal A})$
and let $\tau_1, \ldots, \tau_d$ be all subtrees attached to the root of $X$. 
By Lemma~\ref{P:masp_prop_2}, 
either (i) There exists a decomposition $({\mathcal B}_1, {\mathcal B}_2)$ of ${\mathcal A}$
such that $X$ is an enclosed supertree of ${\mathcal B}_1$; or 
(ii)  There is a decomposition $({\mathcal B}_1, \ldots, {\mathcal B}_d)$ of ${\mathcal A}$
such that each $\tau_j$ is an enclosed supertree of ${\mathcal B}_j$. 
In case (i), we have $|L(X)| \leq \masp({\mathcal B}_1) 
\leq \masp({\mathcal B}_1) + \masp({\mathcal B}_2)$.
On the other hand, in case (ii), we have 
$|L(X)| = |L(\tau_1)| + \ldots + |L(\tau_d)| 
\leq \masp({\mathcal B}_1) + \ldots + \masp({\mathcal B}_d).$
\end{proof}

The above discussion then leads to Theorem~\ref{T:masp_formula}. 

\begin{theorem}
For every sub-forest  ${\mathcal A}$ of $\mathcal T$, 
the value $\masp({\mathcal A})$ equals to
$$
\left\{
\begin{array}{lr}
\left| \Lambda({\mathcal A}) \right|,  \mbox{ if } {\mathcal A} \mbox{ is terminal,}\\
\max \{ \masp({\mathcal B}_1) + \ldots +  \masp({\mathcal B}_d) \mid 
{\mathcal B}_1, \ldots, {\mathcal B}_d \mbox{ decompose } {\mathcal A} \}, 
\mbox{ otherwise} \enspace .
\end{array}
\right.
$$
\label{T:masp_formula}
\end{theorem}

We define an ordering of the sub-forests in $\mathcal O$ as follows. 
For any sub-forests ${\mathcal A}_1, {\mathcal A}_2$ in $\mathcal O$,
we say ${\mathcal A}_1$ is smaller than  ${\mathcal A}_2$ 
if ${\mathcal A}_1^{(i)}$ is either an empty tree or 
a subtree of ${\mathcal A}_2^{(i)}$ for $i=1,2,\ldots, k$.
Our algorithm enumerates ${\mathcal A} \in {\mathcal O}$ in topologically increasing order
and computes $\masp({\mathcal A})$ based on Theorem~\ref{T:masp_formula}.
%
%
%
%
%

In Lemma~\ref{L:masp_decomposition_counting}, we bound 
the number of decompositions of each sub-forest of $\mathcal T$. 
Theorem~\ref{T:masp_complexity} states the complexity 
of the algorithm. 

\begin{lemma}
Each sub-forest of $\mathcal T$ has $O\left( (k D)^{k D + 1} \right)$
decompositions, and generating those decompositions
takes $O\left( k^2D^2 \right)$ time per decomposition. 
\label{L:masp_decomposition_counting}
\end{lemma}	
\begin{proof}
Let ${\mathcal A}$ be a sub-forest of $\mathcal T$.
Since the maximum degree of any agreement supertree of ${\mathcal A}$
is bounded by $k D$, we consider only decompositions that 
consist of at most $k D$ elements. 
We claim that for each $d \in \{2, \ldots, k D \}$, the
sub-forest ${\mathcal A}$ owns $O\left( (d + 2)^{k D} \right)$ decompositions
$({\mathcal B}_1, \ldots, {\mathcal B}_d)$. 
Summing up those asymptotic terms gives us the specified bound. 

The key of this proof is to prove that for each $s \in \{1, \ldots, k \}$, 
the tree ${\mathcal A}^{(s)}$ contributes at most 
$(d + 1)^D + d < (d + 2)^D$ sequences 
${\mathcal B}_1^{(s)}, \ldots, {\mathcal B}_d^{(s)}$, 
and generating those sequences requires $O(d)$ time per sequence. 
We have two cases, each corresponds to a type of the above sequence. 

\textbf{Case 1:} One term in the sequence is ${\mathcal A}^{(s)}$; 
therefore the other terms are empty trees. 
Then, we can generate
this sequence by assigning ${\mathcal A}^{(s)}$ to exactly one term
and setting the rest to be empty trees. This case provides exactly $d$ 
sequences and  enumerates them in $O(d)$ time per sequence. 

\textbf{Case 2:} No term in 
the above sequence is ${\mathcal A}^{(s)}$. Consider an integer
$r \in \{0, 1, \ldots, d \}$ and assume that
the sequence consists of exactly $r$ terms that are nonempty nodes. 
Then those $r$ nonempty trees are isomorphic to pairwise distinct subtrees 
attached to the root of ${\mathcal A}^{(s)}$. 
Let $\delta$ be the degree of the root of ${\mathcal A}^{(s)}$.
We generate the sequence as follows. First we draw $r$ pairwise distinct
subtrees attached to the root of ${\mathcal A}^{(s)}$. 
Next, we select $r$ terms
in the sequence and distribute the above subtrees to them. 
Finally we set the remaining terms to be empty trees. Hence this case 
gives at most 
$$\sum_{r \leq \min\{\delta, d\}} 
\left(
\begin{array}{c}
\delta \\
r
\end{array}
\right) \frac{d!}{(d - r)!}
<           
\sum_{r = 0}^D
\left(
\begin{array}{c}
D \\
r	
\end{array}
\right) d^r = (d + 1)^D$$
sequences, and generates them in $O(d)$ time per sequence. 
\end{proof}

\begin{theorem}
A maximum agreement supertree of $k$ rooted phylogenetic trees 
can be obtained in  $O\left( (k D)^{k D + 3}(2n)^k \right)$ time. 
\label{T:masp_complexity}
\end{theorem}
\begin{proof}
Testing if a sub-forest is terminal takes $O(k)$ times, 
and each terminal sub-forest  ${\mathcal A}$ then requires $O(k^2)$ time 
for computing $\Lambda({\mathcal A})$. By Lemma~\ref{L:masp_decomposition_counting}, 
each non-terminal sub-forest requires $O\left( (k D)^{k D + 3} \right)$ 
running time. Summing up those asymptotic terms for $O\left( (2n)^k \right)$
sub-forests of $\mathcal T$ gives us the specified time bound. 
\end{proof}

\section{Algorithm for MASP of unrooted trees}\label{masp_unroot}

Let ${\mathcal T}$ be a set of $k$ unrooted
phylogenetic trees. This section extends 
the algorithm in Section~\ref{masp}
to find the size of a maximum agreement supertree of $\mathcal T$
in $O\left( (k D)^{k D + 3} (4n)^k \right)$ time. 
The maximum agreement supertree can be obtained by backtracking. 

We say that a set of $k$ rooted trees 
${\mathcal F} = \{ {\mathcal F}^{(1)}, \ldots, {\mathcal F}^{(k)} \}$
is a \emph{rooted variant} of $\mathcal T$ if we can obtain 
each ${\mathcal F}^{(i)}$ by rooting ${\mathcal T}^{(i)}$ 
at some internal node. One naive approach is to use the 
algorithm in the previous section to 
solve MASP for each rooted variant of $\mathcal T$. 
Each rooted variant then gives us a solution, and the maximum of those solutions
is the size of a maximum agreement supertree of $\mathcal T$. 
Because there are $O\left(n^k\right)$ rooted variants of $\mathcal T$, 
this approach adds an $O\left(n^k \right)$ factor to the 
complexity of the algorithm for rooted trees. 

We now show how to improve the above naive algorithm.
As mentioned in the previous section, 
the computation of each rooted variant of $\mathcal T$ consists of
$O\left((2n)^k\right)$ sub-problems which correspond to its sub-forests.
(Please refer to Definition~\ref{D:masp_sub-forest} for 
the concept of sub-forest.)
Since different rooted variants 
may have some common sub-forests, the total number of 
sub-problems we have to run is much smaller than $O(2^k n^{2k})$. 
More precisely, we will show that the total number of
sub-problems is only $O\left((4n)^k\right)$.

A (rooted or unrooted) tree is \emph{trivial} 
if it is a leaf or an empty tree. 
A \emph{maximal subtree} of an unrooted tree $T$ is a rooted tree 
obtained by first rooting $T$ at some internal node $v$
and then removing at most one nontrivial subtree attached to $v$.
Let $\mathcal O$ denote the set of sub-forests of all rooted variants of $\mathcal T$.

\begin{lemma}
Let ${\mathcal A} = \{ {\mathcal A}^{(1)}, \ldots, {\mathcal A}^{(k)} \}$ 
be a set of rooted trees. 
Then ${\mathcal A} \in {\mathcal O}$ if and only if 
each ${\mathcal A}^{(i)}$ is either a trivial subtree 
or a maximal subtree of ${\mathcal T}^{(i)}$. 
\label{L:unroot_masp_characteristic}
\end{lemma}
\begin{proof}
Let $\mathcal F$ be a rooted variant of $\mathcal T$
such that ${\mathcal A}$ is a sub-forest of $\mathcal F$. 
Fix an index $s \in \{1, \ldots, k \}$ and let
$v$ be the root node of ${\mathcal A}^{(s)}$. 
Our claim is straightforward
if either ${\mathcal A}^{(s)}$ is trivial or $v$ is the root node of ${\mathcal F}^{(s)}$.
Otherwise, let $u$ be the parent of $v$ in ${\mathcal F}^{(s)}$. Hence ${\mathcal A}^{(s)}$
is the maximal subtree of ${\mathcal T}^{(s)}$ obtained
by first rooting ${\mathcal T}^{(s)}$ at $v$ and then 
removing the complete subtree rooted at $u$. 

Conversely, we construct a rooted variant $\mathcal F$
of $\mathcal T$ such that 
${\mathcal A}$ is a sub-forest of ${\mathcal F}$ as follows.
For each $i \leq k$,  
if ${\mathcal A}^{(i)}$ is trivial 
or ${\mathcal A}^{(i)}$ is a tree obtained by rooting
${\mathcal T}^{(i)}$ at some internal node
then constructing ${\mathcal F}^{(i)}$ is straightforward. 
Otherwise ${\mathcal A}^{(i)}$ is a maximal subtree of ${\mathcal T}^{(i)}$
obtained by first rooting ${\mathcal T}^{(i)}$ at some internal node
$v$ and then removing exactly one nontrivial subtree $\tau$
attached to $v$. Hence ${\mathcal F}^{(i)}$ is the tree 
obtained by rooting ${\mathcal T}^{(i)}$ at $u$, 
where $u$ is the root of $\tau$. 
\end{proof}

\begin{theorem}
We can find a maximum agreement supertree of $k$
unrooted phylogenetic trees in $O\left( (k D)^{k D + 3} (4n)^k \right)$ time. 
\label{T:unroot_masp_complexity}
\end{theorem}
\begin{proof}
The key of this proof is to show that 
each  tree ${\mathcal T}^{(i)}$ 
contributes at most $(3n - 1)$ maximal subtrees.
It follows that $| {\mathcal O} | \leq (4n)^k$. 
The specified running time of our algorithm is then 
straightforward because each subproblem requires  $O\left( (k D)^{k D + 3} \right)$ 
time as given in the proof of Theorem~\ref{T:masp_complexity}.
Assume that the tree ${\mathcal T}^{(i)}$ has exactly $L$ leaves, 
with $L \leq n$. We now count the number of maximal subtrees $T$
of ${\mathcal T}^{(i)}$ in two cases. 

\textbf{Case 1:} $T$ is obtained by rooting ${\mathcal T}^{(i)}$ 
at some internal node. Hence this case provides
at most $L - 1 < n $ maximal subtrees. 

\textbf{Case 2:} $T$ is obtained by first rooting ${\mathcal T}^{(i)}$
at some internal node $v$ and then removing a nontrivial 
subtree $\tau$ attached to $v$.  
Notice that there is a one-to-one correspondence between the tree $T$
and the directed edge $(v, u)$ of ${\mathcal T}^{(i)}$, 
where $u$ is the root node of $\tau$.  There are $2L -2$ or fewer
undirected edges in ${\mathcal T}^{(i)}$ but exactly $L$ of them 
are adjacent to the leaves. Hence this case gives us at most 
$2 (2L - 2 - L) < 2n - 1$ maximal subtrees. 
\end{proof}

\vskip-0.3cm
\bibliographystyle{plain}

\end{document}